\documentclass[runningheads,a4paper]{llncs}
\pdfoutput=1
\usepackage{amssymb,amsmath}
\usepackage{times}
\usepackage{array}
\usepackage{graphicx}
\usepackage{xcolor}
\usepackage{mathtools}

\usepackage[stable]{footmisc}

\usepackage{supertabular}

\newcommand{\comment}[1]{}

\newcommand{\eqdef}{\stackrel{\mbox{\tiny def}}{=}}
\newcommand{\gp}{\ensuremath{p}}
\newcommand{\gq}{\ensuremath{q}}
\newcommand{\gr}{\ensuremath{r}}
\newcommand{\gs}{\ensuremath{s}}

\makeatletter
\def\provideenvironment{\@star@or@long\provide@environment}
\def\provide@environment#1{%
        \@ifundefined{#1}%
                {\def\reserved@a{\newenvironment{#1}}}%
                {\def\reserved@a{\renewenvironment{#1}}}%
        \reserved@a
}
\def\dummy@environ{}
\makeatother

\provideenvironment{remark}[1][Remark]{\begin{trivlist}
  \item[\hskip \labelsep {\bfseries #1}]}{\end{trivlist}}

\ifx\fmtname\@psfmtname \else \def\cmsy@{2}\fi 
\def\sometime{\mathord{\hbox{\large$\mathchar"0\cmsy@7D$}}}

\newcommand{\always}{\Box}
\newcommand{\nec}[1]{\raisebox{.1ex}{{\ensuremath{\mathmakebox[11pt][c]{\mathclap{\always}\mathclap{\raisebox{.35ex}{\mbox{\tiny$#1$}}}}}}}\hspace{-.16ex}}
\newcommand{\pos}[1]{{\ensuremath{\mathmakebox[11pt][c]{\mathclap{\sometime}\mathclap{\raisebox{.5ex}{\mbox{\tiny$#1$}}}}}}\hspace{-.16ex}}

\newcommand{\then}{\Rightarrow}

\newcommand{\ifonlyif}{\Leftrightarrow}
\newcommand{\tvalue}[1]{\mbox{\it #1\/}}

\newcommand{\constant}[1]{\mbox{\rm\bf #1}}
\newcommand{\true}{\constant{true}}
\newcommand{\false}{\constant{false}}
\newcommand{\start}{\constant{start}}

\newcommand{\system}[3]{\raisebox{.2ex}[1.2ex]{\raisebox{-.2ex}{{\sf #1}}{$_{#2}^{#3}$}}}
\newcommand{\set}[1]{\mbox{$\mathcal{#1}$}}
\newcommand{\Set}[1]{\ensuremath{\{#1\}}}

\newcommand{\agent}{\ensuremath{a}}
\newcommand{\Agents}{\set{A}_n}

\newcommand{\WFF}[2]{{\sf WFF}{\mbox{$_{\mbox{\small\sf #1}_{#2}}$}}}

\newcommand{\relation}[2]{\ensuremath{\mathcal #1_{#2}}}
\newcommand{\universal}{\always^*}
\newcommand{\snf}[1]{{\sf SNF}\mbox{$\textstyle _{#1}$}}

\newcommand{\newpos}[1]{\ensuremath{pos_{#1}}}

\newcommand{\G}[2]{\ensuremath{\axiomname{G}^{#1}_{#2}}}

\newcommand{\Nat}{\mbox{$\mathbb N$}}

\newcommand{\calculus}[1]{{\sf RES}\raisebox{-.5ex}{{\scriptsize{\sf #1}}}}
\newcommand{\rulename}[1]{{\sf\bf #1}}
\newcommand{\axiomname}[1]{{\sf\bf #1}}
\newcommand{\jusdef}[1]{[\rulename{#1}]}
\newcommand{\jus}[2]{[\rulename{#1},#2]}
\newcommand{\res}[2]{\mbox{\rulename{RES}$_{#1}^{#2}$}}

\newcommand{\model}[1]{\ensuremath{\set{#1}}}
\newcommand{\modelw}[2]{\ensuremath{\langle \model{#1}, #2 \rangle}}
\newcommand{\Model}[1]{\Tuple{#1}}

\newcommand{\Sts}{\ensuremath{\set{W}}}
\newcommand{\st}{\ensuremath{w}}

\newcommand{\Tuple}[1]{\ensuremath{(#1)}}

\newenvironment{keyword}{\small \begin{center}{{\bf Keywords: }}}{\end{center}}

\title{Clausal Resolution for Modal Logics of Confluence\footnote{Pre-print version of the paper accepted to IJCAR 2014.}}
\author{Cl\'audia Nalon\inst{1}\thanks{C.\ Nalon was partially supported by CAPES Foundation BEX 8712/11-5.}
\and 
Jo\~ao Marcos\inst{2}\thanks{J.\ Marcos was partially supported by CNPq and by the EU-FP7 Marie Curie project PIRSES-GA-2012-318986.} 
\and Clare Dixon\inst{3}}

\institute{Departament of Computer Science, University of Bras\'{\i}lia \\
           C.P. 4466 -- CEP:70.910-090 -- Bras\'{\i}lia -- DF -- Brazil \\
           \email{nalon@unb.br}
\and           
           LoLITA and Dept.\ of Informatics and Applied Mathematics, UFRN, Brazil \\
           \email{jmarcos@dimap.ufrn.br}
\and
          Department of Computer Science, University of Liverpool \\
          Liverpool,  L69 3BX -- United Kingdom \\
          \email{CLDixon@liverpool.ac.uk}}

\begin{document}
\maketitle

\vspace{-2ex}\begin{abstract}
We present a clausal resolution-based method for normal multimodal logics of confluence, whose Kripke semantics are based on frames characterised by appropriate instances of the Church-Rosser property. Here we restrict attention to eight families of such logics. We show how the inference rules related to the normal logics of confluence can be systematically obtained from the parametrised axioms that characterise such systems. We discuss soundness, completeness, and termination of the method. In particular, completeness can be modularly proved by showing that the conclusions of each newly added inference rule ensures that the corresponding conditions on frames hold. Some examples are given in order to illustrate the use of the method.
\end{abstract}

\section{Introduction}

Modal logics are often introduced as extensions of classical logic with two additional unary operators: ``$\nec{}$'' and ``$\pos{}$'', whose meanings 
vary with the field of application to which they are tailored to apply.
In the most common interpretation, formulae ``$\nec{} p$'' and ``$\pos{} p$'' are read as ``$p$ is necessary'' and ``$p$ is possible'', respectively. Evaluation of a modal formula depends upon an organised collection of scenarios known as {\em possible worlds}. 
Different modal logics assume different {\em accessibility relations} between such worlds. Worlds and their accessibility relations define a so-called {\em Kripke frame}. The evaluation of a formula hinges on such structure: given an appropriate accessibility relation and a world~$\st$, a formula $\nec{} p$ is satisfied at $\st $ if $p$ is true at all worlds accessible from $\st $; a formula $\pos{} p$ is satisfied at $\st $ if $p$ is true at some world accessible from $\st$. 

In normal modal logics extending the classical propositional logic, the schema $\nec{} (\varphi \then \psi) \then (\nec{} \varphi \then \nec{} \psi)$ (the distribution axiom~\axiomname{K}), where $\varphi$ and $\psi$ are well-formed formulae and~$\then$ stands for classical implication, is valid, and the schematic rule $\varphi / \nec{} \varphi$ (the necessitation rule \axiomname{Nec}) preserves validity. The weakest of these logics, named~\system{K}{(1)}{}, is semantically characterised by the class of Kripke frames with no restrictions imposed on the accessibility relation. In the multimodal version, named~\system{K}{(n)}{}, Kripke frames are directed multigraphs and modal operators are equipped with indexes over a set of \emph{agents}, given by $\Agents=\Set{1,2,\ldots,n}$, for some positive integer~$n$. Accordingly, in this case classical logic is extended with operators $\nec{1},\nec{2},\ldots,\nec{n}$, where a formula as $\nec{\agent} p$, with $\agent \in \Agents$, may be read as ``agent $\agent$ considers~$p$ to be necessary''. The modal operator $\pos{\agent}$ is the dual of $\nec{\agent}$, being introduced as an abbreviation for $\neg\nec{\agent}\neg$, where~$\neg$ stands for classical negation. The logic \system{K}{(n)}{} can be seen as the \emph{fusion} of $n$ copies of \system{K}{(1)}{} and its axiomatisation is given by the union of the axioms for classical propositional logic with the axiomatic schemata \axiomname{K$_\agent$}, namely $\nec{\agent} (\varphi \then \psi) \then (\nec{\agent} \varphi \then \nec{\agent} \psi)$, for each $\agent \in \Agents$; and the set of inference rules is given by \emph{modus ponens} and the rule schemata \axiomname{Nec$_\agent$}, namely $\varphi / \nec{\agent} \varphi$, for each $\agent \in \Agents$.

The basic normal multimodal logic \system{K}{(n)}{} and its extensions have been widely used to represent and reason about complex systems. Some of the interesting extensions include the normal multimodal logics based on \axiomname{K$_\agent$} and (the combination of) axioms as, for instance, \axiomname{T$_\agent$} $(\nec{\agent}\varphi \then \varphi)$, \axiomname{D$_\agent$} $(\nec{\agent}\varphi \then\pos{\agent}\varphi)$, \axiomname{4$_\agent$} $(\nec{\agent}\varphi \then \nec{\agent}\nec{\agent}\varphi)$, \axiomname{5$_\agent$} $(\pos{\agent}\varphi \then\nec{\agent}\pos{\agent}\varphi)$, and \axiomname{B$_\agent$} $(\pos{\agent}\nec{\agent}\varphi \then\varphi)$. For example, the description logic {$\mathcal{ALC}$}, which is employed for reasoning about ontologies, is a syntactic variant of \system{K}{(1)}{} \cite{Schild91}; the epistemic logic, denoted by \system{S5}{(n)}{}, which is used in dealing with problems ranging from multi-agency to communication protocols \cite{rao:91c,FHMV95}, can be axiomatised by combining \axiomname{K$_\agent$}, \axiomname{T$_\agent$}, and \axiomname{5$_\agent$}. The addition of those axioms (or their combinations) to \system{K}{(n)}{} imposes some restrictions on the class of models where formulae are valid. Thus, a formula valid in a logic containing \axiomname{T$_\agent$} is valid only if it is valid in a frame where the accessibility relation for each agent $\agent$ is {\em reflexive}. The other axioms, \axiomname{D$_\agent$}, \axiomname{4$_\agent$}, \axiomname{5$_\agent$}, and \axiomname{B$_\agent$}, demand the accessibility relation for each agent $\agent$ to be, respectively, {\em serial}, {\em transitive}, {\em Euclidean}, and {\em symmetric}.

A \emph{logic of confluence} \system{K}{(n)}{\gp,\gq,\gr,\gs} is a modal system axiomatised by  \system{K}{(n)}{} plus axioms $\G{\gp,\gq,\gr,\gs}{\agent}$ of the form 

\[\pos{\agent}^\gp\nec{\agent}^\gq \varphi \then \nec{\agent}^\gr\pos{\agent}^\gs\varphi\]

\noindent where $a\in\Agents$, $\varphi$ is a well-formed formula, $\gp,\gq, \gr,\gs \in \Nat$, where $\nec{\agent}^0\varphi \eqdef \varphi$ and $\nec{\agent}^{i+1}\varphi \eqdef \nec{\agent}\nec{\agent}^{i}\varphi$, and where $\pos{\agent}^0\varphi \eqdef \varphi$ and $\pos{\agent}^{i+1}\varphi \eqdef \pos{\agent}\pos{\agent}^{i}\varphi$, for $i \in \Nat$ (the superscript is often omitted if equal to~$1$). Such axiomatic schemata were notably studied by Lemmon \cite{SL:1977}. Using Modal Correspondence Theory, it can be shown that the frame condition on a logic where an instance of $\G{\gp,\gq,\gr,\gs}{a}$ is valid corresponds to a generalised diamond-like structure representing the \emph{Church-Rosser property} (the philosophical literature sometimes calls such property `incestual' \cite{Chellas80}), as illustrated in Fig.~\ref{fig:diamond} \cite{CarnielliP08}.
To be more precise, let $\Sts$ be a nonempty set of worlds and let $\relation{R}{\agent} \subseteq \Sts \times \Sts$ be the accessibility relation of agent $\agent \in \Agents$.  By $\st\relation{R}{\agent}^0 \st'$ we mean that $\st=\st'$, and $\st\relation{R}{\agent}^{i+1} \st'$ means that there is some world~$\st''$ such that $\st\relation{R}{\agent} \st''$ and $\st''\relation{R}{\agent}^{i} \st'$. Thus, $\st\relation{R}{\agent}^{i} \st'$ holds if there is an $i$-long $\relation{R}{\agent}$-path from~$\st$ to~$\st'$; alternatively, to assert that, we may also write $(\st,\st')\in\relation{R}{\agent}^{i}$.
Given these definitions, the condition on frames that corresponds to the axiom $\G{\gp,\gq,\gr,\gs}{\agent}$ is described by $\forall \st_0,\st_1,\st_2 \; (\st_0\relation{R}{\agent}^\gp \st_1 \land \st_0\relation{R}{\agent}^\gr \st_2 \then \exists \st_3 (\st_1\relation{R}{\agent}^\gq \st_3 \land \st_2 \relation{R}{\agent}^\gs \st_3))$, where $\st_0$, $\st_1$, $\st_2$, $\st_3$ $\in \Sts$.

\begin{figure}[h!]\center
\begin{center}
\includegraphics[width=.5\textwidth]{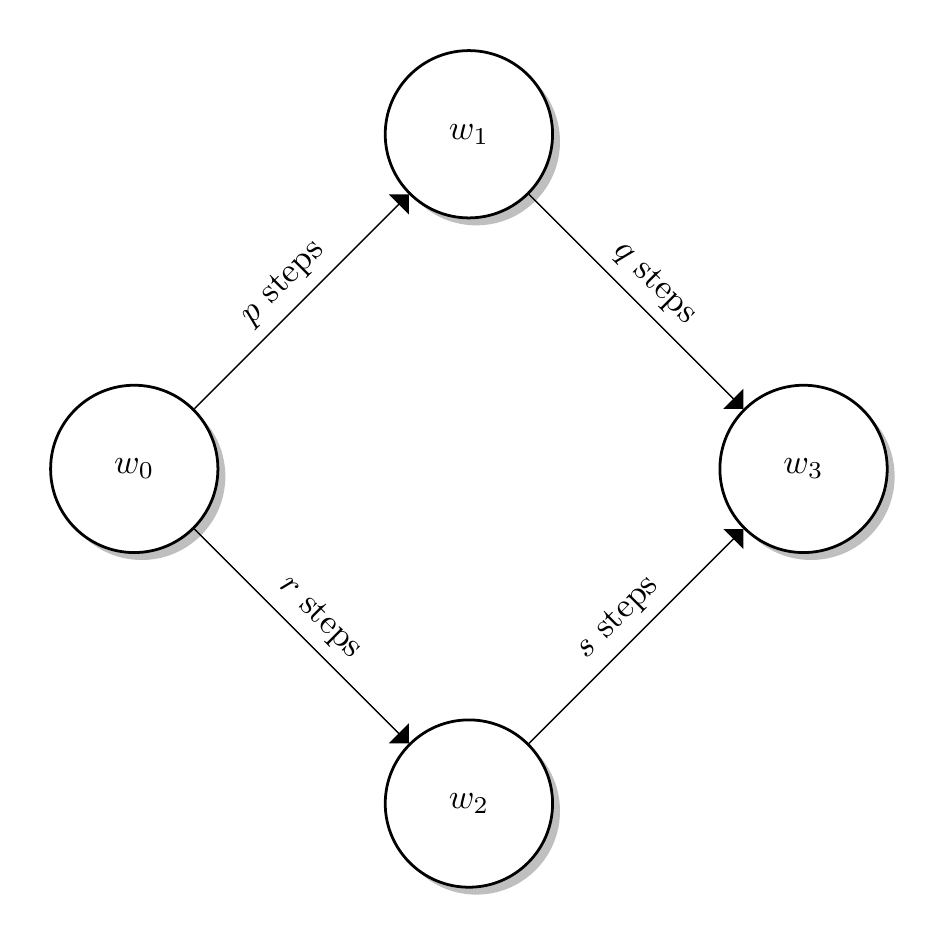}
\end{center}
\caption{Church-Rosser property for frames where $\G{\gp,\gq,\gr,\gs}{a} = \protect\pos{a}^\gp\protect\nec{a}^\gq\varphi \then \protect\nec{a}^\gr\protect\pos{a}^\gs\varphi$ is valid.}
\label{fig:diamond}
\end{figure}

Many well-known modal axiomatic systems are identified with particular logics of confluence. For instance, \system{T}{(n)}{} corresponds to \system{K}{(n)}{0,1,0,0}, a normal modal logic in which the axiom $\nec{\agent}\varphi \then \varphi$ is valid, for all $\agent \in \Agents$ and any formula~$\varphi$. The axiom~\axiomname{4$_\agent$} may be written as $\G{0,1,2,0}{\agent}$, that is, $\nec{\agent}^1 \varphi \then \nec{\agent}^2\varphi$. The Geach axiom \axiomname{G1$_\agent$} is given by $\G{1,1,1,1}{\agent}$ ($\pos{\agent}\nec{\agent}\varphi \then \nec{\agent}\pos{\agent}\varphi$).  
Formulae in \system{K}{(n)}{1,1,1,1} are satisfiable if, and only if, they are satisfiable in a model with $n$ relations satisfying the so-called `diamond property', and analogous claims hold for instance concerning formulae of \system{T}{(n)}{} and models whose relations are all reflexive, and formulae of \system{4}{(n)}{} and models whose relations are all transitive.

Logics of confluence are interesting not only because they encompass a great number of normal modal logics as particular examples, but also in view of their attractive computational behaviour.  Indeed, if we think of multimodal frames as abstract rewriting systems, for instance, and think of modal languages as a way of obtaining an internal and local perspective on such frames, then each given notion of confluence ensures that certain different paths of transformation will eventually lead to the same result. Having a decidable proof procedure for a logic underlying such class of frames helps in establishing a direct form of verifying the properties of the structures that they represent.

As a contribution towards a uniform approach to the development of proof methods for logics of confluence, in this work we deal with the logics where $\gp,\gq,\gr,\gs \in \Set{0,1}$. Table~\ref{table:frames} shows the relevant axiomatic schemata, some standard names by which they are known, and the corresponding conditions on frames. The axiom \axiomname{\G{0,1,1,1}{\agent}} seems not to be named in the literature; the corresponding property follows the naming convention given in \cite[pg.\ 127]{Boolos:lp:1993}.  Note that $\G{0,0,0,0}{\agent}$, $\G{0,1,1,0}{\agent}$, and $\G{1,0,0,1}{\agent}$ are instances of classical tautologies and are thus not included in Table~\ref{table:frames}. Also, given the duality between $\nec{\agent}$ and $\pos{\agent}$, $\G{\gp,\gq,\gr,\gs}{\agent}$ is semantically equivalent to $\G{\gr,\gs,\gp,\gq}{\agent}$. Thus, there are in fact eight families of multimodal logics related to the axioms $\G{\gp,\gq,\gr,\gs}{\agent}$, where $\gp,\gq,\gr,\gs \in \Set{0,1}$. 

{\tiny
\begin{table}[h]
\hspace{-5mm}
\begin{tabular}{|c|c|c|c|p{6.2cm}|}\hline
(p,q,r,s)   &       Name      & Axioms & Property & Condition on Frames \\ \hline
$(0,0,1,1)$ & \axiomname{B$_\agent$}   & $\varphi \then \nec{\agent}\pos{\agent}\varphi$ & symmetric & $\forall w, w' (w \relation{R}{\agent} w' \then w' \relation{R}{\agent} w)$ \\
$(1,1,0,0)$ &                 & $\pos{\agent}\nec{\agent}\varphi \then \varphi$ & & \\ \hline
$(0,0,1,0)$ & \axiomname{Ban$_\agent$} & $\varphi \then \nec{\agent}\varphi$ &  modally banal & $\forall w, w' (w \relation{R}{\agent} w' \then w = w')$  \\
$(1,0,0,0)$ &                 & $\pos{\agent}\varphi \then \varphi$ &  &\\ \hline
$(0,1,0,1)$ & \axiomname{D$_\agent$}   & $\nec{\agent}\varphi \then \pos{\agent}\varphi$ & serial & $\forall w \exists w' (w \relation{R}{\agent} w')$ \\ \hline
$(1,0,1,0)$ & \axiomname{F$_\agent$}   & $\pos{\agent}\varphi \then \nec{\agent}\varphi$ & functional & $\forall w, w', w'' ((w \relation{R}{\agent} w' \land w \relation{R}{\agent} w'') \then w' = w'')$  \\ \hline
$(0,0,0,1)$ & \axiomname{T$_\agent$}   & $\varphi \then \pos{\agent}\varphi$ & reflexive & $\forall w  (w \relation{R}{\agent} w)$ \\
$(0,1,0,0)$ &                 & $\nec{\agent}\varphi \then \varphi$ &           & \\ \hline
$(1,0,1,1)$ & \axiomname{5$_\agent$}   & $\pos{\agent}\varphi \then \nec{\agent}\pos{\agent}\varphi$ & Euclidean & $\forall w, w', w'' ((w \relation{R}{\agent} w' \land w \relation{R}{\agent} w'') \then w' \relation{R}{\agent} w'')$ \\
$(1,1,1,0)$ &                 & $\pos{\agent}\nec{\agent}\varphi \then \nec{\agent}\varphi$ & & \\ \hline
$(1,1,1,1)$ & \axiomname{G1$_\agent$}  & $\pos{\agent}\nec{\agent}\varphi \then \nec{\agent}\pos{\agent}\varphi$ & convergent & $\forall w, w', w'' ((w \relation{R}{\agent} w' \land w \relation{R}{\agent} w'') \then \exists w''' (w' \relation{R}{\agent} w''' \land w'' \relation{R}{\agent} w'''))$ 
\\ \hline
$(0,1,1,1)$ & \axiomname{\G{0,1,1,1}{\agent}} & $\nec{\agent}\varphi \then \nec{\agent}\pos{\agent}\varphi$ & {0,1,1,1-convergent} & 
$\forall w, w' (w \relation{R}{\agent} w' \then \exists w''(w \relation{R}{\agent} w'' \land w' \relation{R}{\agent} w''))$
\\ 
$(1,1,0,1)$ &  & $\pos{\agent}\nec{\agent}\varphi \then \pos{\agent}\varphi$ & {} & \\ \hline
\end{tabular}
\label{table:frames}
\caption{Axioms and corresponding conditions on frames.}
\end{table}
}

We present a clausal resolution-based method for solving the satisfiability problem in logics axiomatised by \axiomname{K$_\agent$} plus $\G{\gp,\gq,\gr,\gs}{\agent}$, where $\gp,\gq,\gr,\gs \in \Set{0,1}$. The resolution calculus is based on that of \cite{ND07:jaal}, which deals with the logical fragment corresponding to $\system{K}{(n)}{}$. The new inference rules to deal with axioms of the form $\G{\gp,\gq,\gr,\gs}{a}$ add relevant information to the set of clauses: the conclusion of each inference rule ensures that properties related to the corresponding conditions on frames hold, that is, the newly added clauses capture the required properties of a model. We discuss soundness, completeness, and termination. Full proofs can be found in \cite{MND13}.

\section{The Normal Modal Logic {\protect\system{K}{(n)}{}}}\label{section-modal-k}


The set \WFF{K}{(n)} of \emph{well-formed formulae} of the logic \system{K}{(n)}{} is constructed from a denumerable set of \emph{propositional symbols}, $\set{P} = \{p, q, p', q', p_1,q_1,\ldots\}$, the negation symbol~$\neg$, the conjunction symbol~$\wedge$, the propositional constant \constant{true}, and a unary connective~$\nec{\agent}$ for each agent~$\agent$ in the finite set of agents $\Agents = \{1, \ldots, n\}$.
When $n=1$, we often omit the index, that is, $\nec{}\varphi$ stands for $\nec{1}\varphi$. As usual, $\pos{\agent}$ is introduced as an abbreviation for $\neg\nec{\agent}\neg$. 
A \emph{literal} is either a propositional symbol or its negation; the set of literals is denoted by~$\set{L}$. By $\neg l$ we will denote the \emph{complement} of the literal~$l\in\set{L}$, that is, $\neg l$ denotes $\neg p$ if $l$ is the propositional symbol~$p$, and $\neg l$ denotes~$p$ if $l$ is the literal $\neg p$. A \emph{modal literal} is either $\nec{\agent}l$ or $\neg\!\,\nec{\agent}l$, where $l \in \set{L}$ and $\agent \in \Agents$. 

We present the semantics of \system{K}{(n)}{}, as usual, in terms of Kripke frames.

\begin{definition} 
  A \emph{Kripke frame $\model{S}$ for $n$ agents over \set{P}} 
  is given by a tuple $\Model{ \Sts, \st_0, \relation{R}{1}, \relation{R}{2}, \ldots, \relation{R}{n}}$, 
  where $\Sts$ is a set of possible {\em worlds} (or {\em states}) with a distinguished world $\st_0\,$\/, 
  and each $\relation{R}{\agent}$ is a binary relation on $\Sts$. 
  A \emph{Kripke model} $\model{M} = \Model{ \model{S},\pi }$ equips a Kripke frame~$\model{S}$
  with a function $\pi:\Sts\rightarrow (\set{P} \rightarrow \{\tvalue{true}, \tvalue{false}\})$ 
  that plays the role of an interpretation that associates to each state $\st\in\Sts$ 
  a truth-assignment to propositional symbols.
\end{definition}

\noindent The so-called accessibility relation~$\relation{R}{\agent}$ is a binary relation that captures the notion of relative possibility from the viewpoint of agent~$\agent$: A pair $(\st,\st')$ is in $\relation{R}{\agent}$ if agent $\agent$ considers world $\st'$ possible, given the information available to her in world~$\st$. We write $\modelw{M}{\st}\models \varphi$ (resp.\ $\modelw{M}{\st}\not\models \varphi$) to say that $\varphi$ is satisfied (resp.\ not satisfied) at the world~$\st$ in the Kripke model $\model{M}$.

\begin{definition} Satisfaction of a formula at a given world~$\st$ of a model~$\model{M}$ is set by:
\begin{itemize}
\item $\modelw{M}{\st} \models \constant{true}$
\item $\modelw{M}{\st} \models p$ if, and only if, $\pi (\st)(p)= \tvalue{true}$, where $p\in\set{P}$
\item$\modelw{M}{\st} \models \neg \varphi$ if, and only if, $\modelw{M}{\st} \not\models \varphi$
\item $\modelw{M}{\st} \models (\varphi \wedge \psi)$ if, and only if, $\modelw{M}{\st} \models \varphi$ and  $\modelw{M}{\st} \models \psi$
\item $\modelw{M}{\st} \models \nec{\agent} \varphi$ if, and only if $\modelw{M}{\st'} \models \varphi$, for all $\st'$ such that $\st\relation{R}{\agent}\st'$
\end{itemize}
\end{definition}

\noindent The formulae $\constant{false}$, $(\varphi \vee \psi)$, $(\varphi \then \psi)$, and $\pos{a}\varphi$ are introduced as the usual abbreviations for $\neg\constant{true}$, $\neg(\neg \varphi \wedge \neg \psi)$, $(\neg \varphi \vee \psi)$, and $\neg \nec{a}\neg\varphi$, respectively. Formulae are interpreted with respect to the distinguished world $\st_0$, that is, satisfiability is defined with respect to pointed-models. 
A formula $\varphi$ is said to be \emph{satisfied in the model $\model{M}=\Model{\model{S},\pi}$ of the Kripke frame $\model{S}=\Model{\Sts,\st_0,\relation{R}{1}, \ldots, \relation{R}{n}}$\/} if $\modelw{M}{\st_0} \models \varphi$; the formula~$\varphi$ is \emph{satisfiable in a Kripke frame $\model{S}$} if there is a model $\model{M}$ of $\model{S}$ such that $\modelw{M}{\st_0} \models \varphi$; and $\varphi$ is said to be \emph{valid in a class $\model{C}$ of Kripke frames} if it is satisfied in any model of any Kripke frame belonging to the class~$\model{C}$.

\section{Resolution for \protect\system{K}{(n)}{}}\label{sec:resolution for K}

In \cite{ND07:jaal}, a sound, complete, and terminating resolution-based method for \system{K}{(n)}{}, which in this paper we call \calculus{K}, is introduced. As the proof-method for logics of confluence presented here relies on \calculus{K}, in order to keep the present paper self-contained, we reproduce the corresponding inference rules here and refer the reader to \cite{ND07:jaal} for a detailed account of the method. The approach taken in the resolution-based method for \system{K}{(n)}{} is clausal: a formula to be tested for (un)satisfiability is first translated into a normal form, explained in Section~\ref{section-modal-snf}, and then the inference rules given in Section~\ref{section-k-rules} are applied until either a contradiction is found or no new clauses can be generated. 

\subsection{A Normal Form for {\protect\system{K}{(n)}{}}}\label{section-modal-snf}

\noindent Formulae in the language of \system{K}{(n)}{} can be transformed into a normal form called Separated Normal Form for Normal Logics (\snf{}). As the semantics is given with respect to a pointed-model, we add a nullary connective $\constant{start}$ in order to represent the world from which we start reasoning. Formally, given a model $\model{M}=\Model{\Sts,\st_0, \relation{R}{1}, \ldots, \relation{R}{n},\pi}$, we have that $\modelw{M}{\st}\models \constant{start}\mbox{ if, and only if, }\st=\st_0$. A formula in \snf{} is represented by a conjunction of clauses, which are true at all reachable states, that is, they have the general form $\bigwedge_i \universal  A_i$, where $A_i$ is a clause and $\universal$, the universal operator, is characterised by (the greatest fixed point of)  $\universal \varphi \ifonlyif \varphi \wedge \bigwedge_{\agent \in \Agents} \nec{\agent} \universal \varphi$, for a formula $\varphi$. Observe that satisfaction of $\universal\varphi$ imposes that~$\varphi$ must hold at the actual world $\st$ and at every world reachable from $\st$, where reachability is defined in the usual (graph-theoretic) way. 
Clauses have one of the following forms: 

{\small	
\begin{itemize}
\item \parbox{2in}{Initial clause}\parbox{2in}{\[\constant{start} \then \bigvee_{b=1}^{r} l_b\]}

\item \parbox{2in}{Literal clause}\parbox{2in}{\[\constant{true} \then \bigvee_{b=1}^{r} l_b\]}

\item \parbox{2in}{Positive $\agent$-clause}\parbox{2in}{\[l' \then \nec{\agent}l\]}

\item \parbox{2in}{Negative $\agent$-clause}\parbox{2in}{\[l' \then \neg \nec{\agent}l\]}
\end{itemize}
}

\noindent where $l$, $l'$, $l_b \in \set{L}$. Positive and negative $\agent$-clauses are together known as {\em modal~$\agent$-clauses}; the index~$\agent$ may be omitted if it is clear from the context. 

The translation to \snf{} uses rewriting of classical operators and the renaming technique \cite{PG86}, where complex subformulae are replaced by new propositional symbols and the truth of these new symbols is linked to the formulae that they replaced in all states. Given a formula $\varphi$, the translation procedure is applied to $\universal(\start \then t_0) \land \universal(t_0 \then \varphi)$, where $t_0$ is a new propositional symbol. 
The universal operator, which surrounds all clauses, ensures that the clauses generated by the translation of a formula are true at all reachable worlds. 
Classical rewriting is used to remove some classical operators from~$\varphi$ (e.g.\  $\universal(t \then \psi_1 \land \psi_2)$ is rewritten as $\universal(t \then \psi_1) \land \universal(t \then \psi_2)$). Renaming is used to replace complex subformulae in disjunctions (e.g.\  if $\psi_2$ is not a literal, $\universal(t \then \psi_1 \lor \psi_2)$ is rewritten as $\universal(t \then \psi_1 \lor t_1) \land \universal(t_1 \then \psi_2)$, where~$t_1$ is a new propositional symbol) or in the scope of modal operators (e.g.\  if $\psi$ is not a literal, $\universal(t \then \nec{\agent}\psi)$ is rewritten as $\universal(t \then \nec{\agent}t_1) \land \universal(t_1 \then \psi)$, where $t_1$ is a new propositional symbol). We refer the reader to~\cite{ND07:jaal} for details on the transformation rules that define the translation to \snf{}, their correctness, and examples of their application.

\subsection{Inference Rules for {\protect\system{K}{(n)}{}}}\label{section-k-rules}

In the following, $l$, $l'$, $l_i$, $l'_i \in \set{L}$ ($i \in \Nat$) and $D$, $D'$ are disjunctions of literals. 

\vspace{.1ex}{\noindent\bf Literal Resolution.} This is classical resolution applied to the classical propositional fragment of the combined logic. An initial clause may be resolved with either a literal clause or another initial clause (rules \rulename{IRES1} and \rulename{IRES2}). Literal clauses may be resolved together (\rulename{LRES}).

\noindent{\small
\begin{center}
\setlength{\arraycolsep}{1pt}
\begin{tabular}{ccc}
$
\begin{array}{lrcl}
\mbox{[\rulename{IRES1}]} &\universal(\constant{true} & \then & D \vee l)\\
 &\universal(\constant{start} & \then  & D' \vee \neg l)\\ \cline{2-4}
 &\universal(\constant{start} & \then & D \vee D')
\end{array}
$
&
$
\begin{array}{lrcl} 
\mbox{[\rulename{IRES2}]} &\universal(\constant{start} & \then &  D \vee l)\\
 &\universal(\constant{start} & \then & D' \vee \neg l)\\  \cline{2-4}
 &\universal(\constant{start} & \then & D \vee D')
\end{array}
$
&
$
\begin{array}{lrcl}
\mbox{[\rulename{LRES}]} &\universal(\constant{true} & \then &  D \vee l)\\
 &\universal(\constant{true} & \then & D' \vee \neg l)\\  \cline{2-4}
 &\universal(\constant{true} & \then & D \vee D')
\end{array}
$
\end{tabular}
\end{center}}
 
\vspace{1ex}{\noindent\bf Modal Resolution.} These rules are applied
between clauses which refer to the same context, that is, they must refer
to the same agent. For instance, we may resolve two or more
$\nec{\agent}$-clauses (rules \rulename{MRES} and \rulename{NEC2}); or 
several $\nec{\agent}$-clauses and a literal clause (rules \rulename{NEC1} and \rulename{NEC3}). 
The modal inference rules are: 

\label{modal-resolution}

{\small
\begin{center}
\setlength{\arraycolsep}{1pt}
\begin{tabular}{@{\hspace{-5pt}}l@{\hspace{-15pt}}c}

\begin{tabular}{l}
$
\begin{array}{lrcl}
\mbox{[\rulename{MRES}]} &\universal( l_1  & \then & \nec{\agent} l) \\
              &\universal( l_2 & \then & \neg \nec{\agent} l)\\ \cline{2-4}
              &\universal( \constant{true} & \then & \neg l_1 \vee \neg l_2)
\\
\\
\end{array}
$
\\[2mm]
$
\begin{array}{lrcl}
\mbox{[\rulename{NEC1}]} & \universal(l'_1 & \then & \nec{\agent}\neg l_1 )\\[-1mm]
             & & \vdots & \\[-1mm]
             & \universal(l'_m & \then & \nec{\agent}\neg l_m )\\
             & \universal(l' & \then & \neg \nec{\agent} l )\\
             & \universal(\constant{true} & \then & l_1 \vee \ldots \vee l_m \vee l) \\  \cline{2-4}
             & \universal(\constant{true} & \then & \neg l'_1 \vee \ldots \vee \neg l'_m \vee \neg l')
\end{array} 
$
\end{tabular}
&
\begin{tabular}{l}
$
\begin{array}{lrcl}
\mbox{[\rulename{NEC2}]} & \universal(l'_1 & \then & \nec{\agent} l_1) \\
              & \universal(l'_2 & \then & \nec{\agent} \neg l_1) \\
              & \universal(l'_3 & \then & \neg \nec{\agent} l_2) \\ \cline{2-4}
              & \universal(\constant{true} & \then & \neg  l'_1 \vee  \neg l'_2 \vee \neg l'_3)
\end{array} 
$
\\[2mm]
$
\begin{array}{lrcl}
 \mbox{[\rulename{NEC3}]} & \universal(l'_1 & \then & \nec{\agent}\neg l_1 )\\[-1mm]
             & & \vdots & \\[-1mm]
             & \universal(l'_m & \then & \nec{\agent}\neg l_m )\\
             & \universal(l' & \then & \neg \nec{\agent} l )\\
             & \universal(\constant{true} & \then & l_1 \vee \ldots \vee
 l_m) \\  \cline{2-4}
             & \universal(\constant{true} & \then & \neg l'_1 \vee \ldots
 \vee \neg l'_m \vee \neg l')
 \end{array}
 $
\end{tabular}
\end{tabular}
\end{center}}

\noindent 
The rule \rulename{MRES} is a syntactic variation of classical resolution, as a formula and its negation cannot be true at the same state. The rule \rulename{NEC1} corresponds to necessitation (applied to $(\neg l_1 \land \ldots \land \neg l_m \then \neg l)$, which is equivalent to the literal clause in the premises) and several applications of classical resolution. The rule \rulename{NEC2} is a special case of \rulename{NEC1}, as the parent clauses can be resolved with the tautology $\constant{true} \then l_1 \vee \neg l_1 \vee l_2$. The rule \rulename{NEC3} is similar to \rulename{NEC1}, however the negative modal clause is not resolved with the literal clause in the premises. Instead, the negative modal clause requires that resolution takes place between literals on the right-hand side of positive modal clauses and the literal clause. The resolvents in the inference rules \rulename{NEC1}--\rulename{NEC3} impose that the literals on the left-hand side of the modal clauses in the premises are not all satisfied whenever their conjunction leads to a contradiction in a successor state. Given the syntactic forms of clauses, the three rules are needed for completeness \cite{ND07:jaal}. Note that for \rulename{NEC1}, we may have $m=0$; for \rulename{NEC2} the number of premises is fixed; and that for \rulename{NEC3}, if $m=0$, then the literal clause in the premises is $\true \then \false$, which cannot be satisfied in any model. Thus, \rulename{NEC3} is not applied when $m=0$.

We define a derivation as a sequence of sets of clauses $\set{T}_0$, \set{T}$_1$, \ldots, where \set{T}$_i$ results from adding to \set{T}$_{i-1}$ the resolvent obtained by an application of an inference rule of \calculus{K} to clauses in \set{T}$_{i-1}$. A derivation \emph{terminates} if, and only if, either a contradiction, in the form of $\universal(\start \then \false)$ or $\universal(\true \then \false)$, is derived or no new clauses can be derived by further application of the resolution rules of $\calculus{K}$. We assume standard simplification from classical logic to keep the clauses as simple as possible. For example, $D \lor l \lor l$ on the right-hand side of a clause would be rewritten as $D \lor l$. 
 
\begin{example}\label{example:resolution:kn} We wish to check whether the formula $\nec{1}\nec{2}(a \land b) \then \nec{1}(\nec{2}a \land \nec{2}b)$ is valid in \system{K}{(2)}{}. The translation of its negation into the normal form is given by clauses (1)--(9) below. Then the inference rules are applied until $\constant{false}$ is generated. In order to improve readability, the universal operator is suppressed. The full refutation follows.

{\small
\begin{center}
$
\begin{array}{lrlll}
1. &  \start & \then & t_1 \\                                             
2. &  t_1 & \then &\nec{1}t_2  \\                                          
3. &  t_2 & \then &\nec{2}t_3  \\                                          
4. &  \true & \then & \neg t_3\lor a \\                                         
5. &  \true & \then & \neg t_3\lor b \\                                         
6. &  t_1 & \then & \neg \nec{1} \neg t_4 \\                                          
7. &  \true & \then & \neg t_4\lor t_5\lor t_6 \\
8. &  t_5 & \then & \neg \nec{2}a        \\                                    
\end{array}
$\quad\quad
$
\begin{array}{lrlll}                                       
9. &  t_6 & \then & \neg \nec{2}b        \\                          
10. &  \true & \then & \neg t_2\lor \neg t_5  &  [\rulename{NEC1}, 3,8,4]\\
11. &  \true & \then & \neg t_2\lor \neg t_4 \lor t_6 & [\rulename{LRES}, 10,7]\\
12. &  \true & \then & \neg t_2\lor \neg t_6 & [\rulename{NEC1}, 3,9,5]\\
13. &  \true & \then & \neg t_2\lor \neg t_4 & [\rulename{LRES}, 12,11]\\
14. &  \true & \then & \neg t_1 & [\rulename{NEC1}, 2,6,13]\\
15. &  \start & \then & \false & [\rulename{IRES1},14,1]\\
\\
\end{array}
$
\end{center}}

\noindent Clauses (10) and (12) are obtained by applications of \rulename{NEC1} to clauses in the context of agent 2. Clause (14) is obtained by an application of the same rule, but in the context of agent 1. Clauses (11) and (13) result from applications of resolution to the propositional part of the language shared by both agents. Clause (15) shows that a contradiction was found at the initial state. Therefore, the original formula is valid.

\end{example}

\section{Clausal Resolution for Logics of Confluence}\label{sec:calculus:convergence}

The inference rules of \calculus{K}, given in Section~\ref{section-k-rules}, are resolution-based: whenever a set of (sub)formulae is identified as contradictory, the resolvents require that they are not all satisfied together. The extra inference rules for $\system{K}{(n)}{\gp,\gq,\gr,\gs}$, with $\gp,\gq,\gr,\gs \in \Set{0,1}$, which we are about to present, have a different flavour: whenever we can identify that the set of clauses imply that $\pos{a}^\gp\nec{a}^\gq \psi$ holds, we add some new clauses that ensure that $\nec{a}^\gr\pos{a}^\gs \psi$ also holds. If this is not the case, that is, if the set of clauses implies that $\neg\nec{a}^\gr\pos{a}^\gs \psi$ holds, then a contradiction is found by applying the inference rules for  $\system{K}{(n)}{}$. 
Because of the particular normal form we use here, there are, in fact, two general forms for the inference rules for $\system{K}{(n)}{\gp,\gq,\gr,\gs}$, given in Table~\ref{table:rules:general} (where $l,l'$ are literals and $C$ is a conjunction of literals).\vspace{-2ex}

{\small
\begin{table}[h]
\centering
\begin{tabular}{|c|c|}
\hline
$
\begin{array}{lrcl}
{[\res{\agent}{\gp,1,\gr,\gs}]} & \universal(l & \then & \nec{\agent} l')\\ \cline{2-4}
                      & \universal(\pos{\agent}^{\,\gp} l & \then & \nec{\agent}^\gr\pos{\agent}^{\gs} l')
\end{array}
$
\quad\quad &
$
\begin{array}{lrcl}
{[\res{\agent}{\gp,0,\gr,\gs}]} & \universal(C & \then & \pos{\agent}^{\,\gp} l') \\ \cline{2-4}
                     & \universal(C & \then & \nec{\agent}^\gr\pos{\agent}^{\gs} l')
\end{array}
$
\\ \hline
\multicolumn{2}{c}{}\\
\end{tabular}
\caption{Inference Rules for \protect{$\G{\gp,\gq,\gr,\gs}{\agent}$}}
\label{table:rules:general}
\end{table}}

\vspace{-4ex}Soundness is checked by showing that the transformation of a formula $\varphi \in  \WFF{K}{(n)}$ into its normal form is satisfiability-preserving and that the application of the inference rules are also satisfiability-preserving. Satisfiability-preserving results for the transformation into \snf{} are provided in \cite{ND07:jaal}. To extend the soundness results so as to cover the new inference rules, note that the conclusions of the inference rules in Table~\ref{table:rules:general} are derived using the semantics of the universal operator and the distribution axiom, $\axiomname{K}_{\agent}$. For \res{\agent}{\gp,1,\gr,\gs}, we have that the premise $\universal(l  \then  \nec{\agent} l')$ is semantically equivalent to $\universal(\neg \nec{\agent} l' \then  \neg l)$. By the definition of the universal operator, we obtain $\universal(\nec{\agent}^{\gp} (\neg \nec{\agent} l'  \then  \neg l))$. Applying the distribution axiom \axiomname{K}$_{\agent}$ to this clause results in $\universal(\nec{\agent}^{\gp} \neg \nec{\agent} l'  \then  \nec{\agent}^{\gp}\neg l))$, which is semantically equivalent to $\universal(\neg \nec{\agent}^{\gp}\neg l  \then  \neg\nec{\agent}^{\gp} \neg \nec{\agent} l')$. As $\pos{\agent}$ is an abbreviation for $\neg\nec{\agent}\neg$ and because $\pos{\agent}^{\gp} \nec{\agent} l'$ implies $\nec{\agent}^{\gr} \pos{\agent}^{\gs} l'$ in $\system{K}{(n)}{\gp,1,\gr,\gs}$, by classical reasoning, we have that $\universal(\neg \nec{\agent}^{\gp}\neg l  \then  \neg\nec{\agent}^{\gp} \neg \nec{\agent} l')$ implies $\universal(\pos{\agent}^{\gp}l  \then  \nec{\agent}^{\gr} \pos{\agent}^{\gs} l')$, the conclusion of \res{\agent}{\gp,1,\gr,\gs}. Soundness of the inference rule \res{\agent}{\gp,0,\gr,\gs} can be proved in a similar way. 

As the conclusions of the above inference rules may contain complex formulae, they might need to be rewritten into the normal form. Thus, we also need to add clauses corresponding to the normal form of $\pos{a}^\gp l$ and $\pos{a}^\gs l'$, which occur in the conclusions of the inference rules. Let $\varphi$ be a formula and let $\tau(\varphi)$ be the set of clauses resulting from the translation of $\varphi$ into the normal form. Let $\set{L}(\tau(\varphi))$ be the set of literals that might occur in the clause set, that is, for all $p \in \set{P}$ such that $p$ occurs in $\tau(\varphi)$, we have that both $p$ and $\neg p$ are in $\set{L}(\tau(\varphi))$. The set of \emph{definition clauses} is given by\vspace{-2ex}

\[
\begin{array}{rcl}
\universal(\newpos{a,l} & \then & \neg \nec{a}\neg l)\\
\universal(\neg \newpos{a,l} & \then & \nec{a}\neg l)\\
\end{array}
\]

\noindent for all $l \in \set{L}(\tau(\varphi))$, where $\newpos{a,l}$ is a new propositional \emph{definition symbol} used for renaming the negative modal literal $\pos{a} l$, that is, the definition clauses correspond to the normal form of $\newpos{a,l}\ifonlyif\neg\nec{a}\neg l $. Note that we have definition clauses for every propositional symbol and its negation, e.g.\ for a propositional symbol $p \in \tau(\varphi)$, we have the definition clauses $\universal(\newpos{a,p} \then \neg \nec{a}\neg p)$, $\universal(\neg \newpos{a,p} \then \nec{a}\neg p)$, $\universal(\newpos{a,\neg p} \then \neg \nec{a} p)$, and  $\universal(\neg \newpos{a,\neg p} \then \nec{a}p)$, for every $\agent \in \Agents$ occurring in $\tau(\varphi)$. We assume the set of definition clauses to be available whenever those symbols are used. It is also important to note that those new definition symbols and the respective definition clauses can all be introduced at the beginning of the application of the resolution method because we do not need definition clauses applied to definition symbols in the proofs, as given in the completeness proof \cite{MND13}. As no new propositional symbols are introduced by the inference rules, there is a finite number of clauses that might be expressed (modulo simplification) and, therefore, the clausal resolution method for each modal logic of confluence is terminating.

As discussed above and from the results in \cite{ND07:jaal}, we can establish the soundness of the proof method.

\begin{theorem}\label{theo:soundness:confluence} The resolution-based calculi for logics of confluence are sound.
\end{theorem}

\begin{proof}[Sketch]
The transformation into the normal form is satisfiability preserving \cite{ND07:jaal}. Given a set \set{T} of clauses and a model \model{M} that satisfies \set{T}, we can construct a model \model{M'} for the union of \set{T} and the definition clauses, where \model{M} and \model{M'} may differ only in the valuation of the definition symbols. By setting properly the valuations in \model{M'}, we have that $\modelw{M'}{\st}\models \newpos{\agent,p}$ if and only if $\modelw{M}{\st}\models \pos{\agent}p$, for any $\st \in \Sts$. Soundness of the inference rules for \calculus{K} is also given in \cite{ND07:jaal}. Soundness of \res{\agent}{\gp,1,\gr,\gs} and \res{\agent}{\gp,0,\gr,\gs} follow from the axiomatisation of \system{K}{(n)}{\gp,\gq,\gr,\gs}.
\end{proof}

\setlength{\textfloatsep}{0.1cm}
{\tiny
\begin{table}
\begin{tabular}{|c|l|}\hline
       Logic      &  Inference Rules \\ \hline

\axiomname{T$_{\agent}$}   & {\footnotesize
$
\begin{array}{lrcl}
\mbox{[\res{\agent}{0,0,0,1}]} & \universal(\constant{true} & \then & D \lor l) \\ \cline{2-4}
& \universal(\neg D & \then & \neg\nec{a}\neg l)
\end{array}
$
}\\
                 & {\footnotesize
$
\begin{array}{lrcl}
\mbox{[\res{\agent}{0,1,0,0}]} & \universal(l & \then & \nec{\agent}l') \\  \cline{2-4}
& \universal(\constant{true} & \then & \neg l \vee l')
\end{array}
$
}\\ \hline
\axiomname{Ban$_{\agent}$} & {\footnotesize
$
\begin{array}{lrcl}
\mbox{[\res{\agent}{0,0,1,0}]} & \universal(\constant{true} & \then & D \lor l) \\ \cline{2-4}
& \universal(\neg D & \then & \nec{a} l)
\end{array}
$
}\\
                 & {\footnotesize
$
\begin{array}{lrcl}
\mbox{[\res{\agent}{1,0,0,0}]} & \universal(l & \then & \neg\nec{\agent}\neg l') \\ \cline{2-4}
& \universal(\constant{true} & \then & \neg l \lor  l')
\end{array}
$
}\\ \hline

 \axiomname{B$_{\agent}$}   & {\footnotesize
$
\begin{array}{lrcl}
\mbox{[\res{\agent}{0,0,1,1}]} & \universal(\constant{true} & \then & D \lor l) \\ \cline{2-4}
& \universal(\neg D & \then & \nec{a} \newpos{a,l})
\end{array}
$
}\\
                 & {\footnotesize
$
\begin{array}{lrcl}
\mbox{[\res{\agent}{1,1,0,0}]} & \universal(l & \then & \nec{\agent} l') \\ \cline{2-4}
&  \universal(\neg l' & \then & \nec{\agent} \neg l)
\end{array}
$
}\\ \hline

 \axiomname{D$_{\agent}$}   & {\footnotesize
$
\begin{array}{lrcl}
\mbox{[\res{\agent}{0,1,0,1}]} & \universal(l & \then & \nec{\agent}l') \\ \cline{2-4}
              & \universal(l & \then & \neg \nec{\agent}\neg l') \\     
\end{array}
$
}\\\hline
\end{tabular}\quad
\begin{tabular}{|c|l|}\hline
       Logic      &  Inference Rules \\ \hline
 \axiomname{\G{0,1,1,1}{\agent}} & {\footnotesize
$
\begin{array}{lrcl}
\mbox{[\res{\agent}{0,1,1,1}]} & \universal(l & \then & \nec{\agent} l') \\ \cline{2-4}
             & \universal(l & \then & \nec{\agent}\newpos{a,l'}) \\ 
\end{array}
$
}\\ 
   & {\footnotesize
$
\begin{array}{lrcl}
\mbox{[\res{\agent}{1,1,0,1}]} & \universal(l & \then & \nec{\agent} l') \\ \cline{2-4}
             & \universal(\newpos{\agent,l} & \then & \neg\nec{\agent}\neg l') \\ 
\end{array}
$
}\\ \hline
 \axiomname{F$_{\agent}$}  &{\footnotesize
$
\begin{array}{lrcl}
\mbox{[\res{\agent}{1,0,1,0}]} & \universal(l & \then & \neg\nec{\agent}\neg l') \\ \cline{2-4}
& \universal(l & \then & \nec{\agent}l')
\end{array}
$
}  \\ \hline

 \axiomname{5$_{\agent}$}   &{\footnotesize
$
 \begin{array}{lrcl}
 \mbox{[\res{\agent}{1,0,1,1}]} & \universal(l & \then & \neg \nec{\agent}\neg l') \\ \cline{2-4}
 & \universal(l & \then & \nec{\agent} \newpos{\agent,l'})\\
 \end{array}$
}\\
                 & {\footnotesize
$\begin{array}{lrcl}
\mbox{[\res{\agent}{1,1,1,0}]} &\universal( l  & \then &  \nec{\agent} l')\\ \cline{2-4}
              & \universal(\newpos{\agent,l} & \then & \nec{\agent} l')\\
\end{array}
$
}\\ \hline
 \axiomname{G1$_{\agent}$}  & {\footnotesize
$
\begin{array}{lrcl}
\mbox{[\rulename{\res{\agent}{1,1,1,1}}]} 
              &\universal( l & \then & \nec{\agent} l') \\ \cline{2-4}
              &\universal( \newpos{\agent,l} & \then & \nec{\agent} \newpos{\agent,l'}) \\
          \end{array}
$
}
\\ \hline
\multicolumn{2}{c}{}\\
\multicolumn{2}{c}{}\\
\end{tabular}
\vspace{1ex}\caption{Inference Rules for several instances of \protect{$\G{\gp,\gq,\gr,\gs}{\agent}$}}
\label{table:rules:specific}
\end{table}
}

Table~\ref{table:rules:specific} 
shows the inference rules for each specific instance of $\G{\gp,\gq,\gr,\gs}{\agent}$, where $\gp,\gq,\gr,\gs \in \Set{0,1}$, $l$, $l' \in \set{L}$, and $D$ is a disjunction of literals. As $\G{\gp,\gq,\gr,\gs}{\agent}$ is semantically equivalent to $\G{\gr,\gs,\gp,\gq}{\agent}$, the inference rules for both systems are grouped together. Some of the inference rules in Table~\ref{table:rules:specific} are obtained directly from Table~\ref{table:rules:general}. For instance, the rule for reflexive systems, i.e.\ where the axiom $\G{0,1,0,0}{\agent}$ is valid, has the form $\universal(l \then \nec{\agent} l') / \universal(\pos{\agent}^{0} l  \then \nec{\agent}^{0}\pos{\agent}^{0} l')$ in Table~\ref{table:rules:general}; in Table~\ref{table:rules:specific}, the conclusion is rewritten in its normal form, that is, $\universal(\true \then \neg l \lor l')$. For other systems, the form of the inference rules are slightly different from what would be obtained from a direct application of the general inference rules in Table~\ref{table:rules:general}. This is the case, for instance, for the inference rules for symmetric systems, that is, those systems where the axiom $\G{1,1,0,0}{\agent}$ is valid. From Table~\ref{table:rules:general}, in symmetric systems, for a premise of the form $\universal(l \then \nec{\agent} l')$, the conclusion is given by $\universal(\pos{\agent} l \then l')$, which is translated into the normal form as $\universal(\true \then \neg\newpos{\agent,l} \vee l')$. We have chosen, however, to translate the conclusion as $\universal(\neg l' \then \nec{\agent} \neg l)$, which is semantically equivalent to the conclusion obtained by the general inference rule, but avoids the use of definition symbols.

The inference rules given in Table~\ref{table:rules:general} provide a systematic way of designing the inference rules for each specific modal logic of confluence. We note, however, that we do not always need both inference rules in order to achieve a complete proof method for a particular logic. In the completeness proofs provided in \cite{MND13}, we show for instance that the inference rules which introduce modalities in their conclusions from literal clauses (that is, the inference rules \res{\agent}{0,0,\gr,\gs}) are not needed for completeness. We also show that we need just one specific inference rule for logics in which \axiomname{\G{0,1,1,1}{\agent}} and \axiomname{5$_\agent$} are valid: $\res{\agent}{0,1,1,1}$ and $\res{\agent}{1,0,1,1}$, respectively.

Given a formula $\varphi$ in $\system{K}{(n)}{\gp,\gq,\gr,\gs}$, with $\gp,\gq,\gr,\gs \in \Set{0,1}$, the resolution method for~$\system{K}{(n)}{}$, given in Section~\ref{sec:resolution for K}, and the inference rule $\res{\agent}{\gp,\gq,\gr,\gs}$ are applied to $\tau(\varphi)$ and the set of definition clauses. The extra inference rules for $\system{K}{(n)}{\gp,\gq,\gr,\gs}$ do not need to be applied to clauses if such application generates new nested definition symbols, that is, we do not need definition clauses for definition symbols. For instance, the application of $\rulename{\res{\agent}{1,1,1,1}}$ to a clause of the form $\universal( l \then \nec{\agent} \newpos{a,l'})$ would result in $\universal( \newpos{\agent,l} \then \nec{\agent} \newpos{\agent,\newpos{a,l'}})$. Although it is not incorrect to apply the inference rules to such a clause, this might cause the method not to terminate. We can show, however, that the application of inference rules to clauses which would result in nested literals is not needed for completeness, as the restrictions imposed by those symbols are already ensured by existing definition symbols and relevant inference rules (see Theorem~\ref{theo:termination:confluence} below). This ensures that no new definition symbols are introduced by the proof method. 

Completeness is proved by showing that, for each specific logic of confluence, if a given set of clauses is unsatisfiable, there is a refutation produced by the method presented here. The proof is by induction on the number of nodes of a graph, known as \emph{behaviour graph}, built from a set of clauses. The graph construction is similar to the construction of a canonical model, followed by filtrations based on the set of formulae (or clauses), often used to check completeness for proof methods in modal logics (see \cite{blackburn_p-etal:2001a}, for instance, for definitions and examples). Intuitively, nodes in the graph correspond to states and are defined as maximally consistent sets of literals and modal literals occurring in the set of clauses, including those literals introduced by definition clauses. That is, for any literal $l$ occurring in the set of clauses, including definition clauses, and agents $\agent \in \Agents$, a node contains either $l$ or $\neg l$; and either $\nec{\agent}l$ or $\neg\nec{\agent} l$. The set of edges correspond to the agents' accessibility relations. Edges or nodes that do not satisfy the set of clauses are deleted from the graph. Such deletions correspond to applications of one or more of the inference rules. We prove that an empty behaviour graph corresponds to an unsatisfiable set of clauses and that, in this case, there is a refutation using the inference rules for \calculus{K}, given in Section~\ref{sec:resolution for K}, and the inference rules for the specific logic of confluence, presented in Table~\ref{table:rules:specific}.

\begin{theorem}\label{theo:completeness:confluence} Let \set{T} be an unsatisfiable set of clauses in \G{\gp,\gq,\gr,\gs}{a}, with $\gp,\gq,\gr,\gs \in \Set{0,1}$. A contradiction can be derived by applying the resolution rules for \calculus{K}, presented in Section~\ref{sec:resolution for K}, and Table~\ref{table:rules:specific}.
\end{theorem}

\begin{proof}[Sketch] We construct a behaviour graph and show that the application of rules in Table~\ref{table:rules:specific} removes nodes and edges where the corresponding frame condition does not hold. The full proof is provided in \cite{MND13}.
\end{proof}

\begin{theorem}\label{theo:termination:confluence} The resolution-based calculi for logics of confluence terminate.
\end{theorem}

\begin{proof}[Sketch] From the completeness proof, the introduction of a literal such as \newpos{\agent,\newpos{\agent,l}} for an agent $\agent$ and literal $l$ is not needed. We can show that the restrictions imposed by such clauses, together with the resolution rules for each specific logical system, are enough to ensure that the corresponding frame condition already holds. As the proof method does not introduce new literals in the clause set, there is only a finite number of clauses that can be expressed. Therefore, the proof method is terminating.
\end{proof}

\begin{example} We show that $\varphi \eqdef p \then \nec{1}\pos{1}p$, which is an instance of \axiomname{B$_1$}, is a valid formula in symmetric systems. As symmetry is implied by reflexivity and Euclideanness, instead of using \res{1}{1,1,0,0}, we combine the inference rules for both \axiomname{T$_1$} and \axiomname{5$_1$}. Clauses (1)--(4) correspond to the translation of the negation of $\varphi$ into the normal form. Clauses (5)--(8) are the definition clauses used in the proof.

{\small
\hspace{-4ex}$
\begin{array}{rrclll}
1. & \start & \then & t_0 \\
2. & \true  & \then & \neg t_0 \lor p \\
3. & t_0 & \then & \neg \nec{1}\neg t_1 \\
4. & t_1 & \then & \nec{1} \neg p \\
5. & \neg \newpos{1,t_1} & \then & \nec{1} \neg t_1 & \jusdef{Def.\ \mbox{$\newpos{1,t_1}$}}\\
6. & \newpos{1,t_1} & \then & \neg\nec{1}\neg t_1 & \jusdef{Def.\ \mbox{$\newpos{1,t_1}$}}\\
7. & \newpos{1,p} & \then & \neg \nec{1} \neg p & \jusdef{Def.\ \mbox{$\newpos{1,p}$}}\\
8. &  \neg \newpos{1,p} & \then &\nec{1} \neg p & \jusdef{Def.\ \mbox{$\newpos{1,p}$}}\\
\\
\end{array}
$\quad
$
\begin{array}{rrclll}
9. & \true & \then &  \neg t_0 \lor \newpos{1,t_1}  & \jus{MRES}{5,3}\\
10. & \true & \then & \neg t_1 \lor \neg \newpos{1,p} & \jus{MRES}{7,4} \\
11. & \newpos{1,p} & \then & \nec{1} \newpos{1,p} & \jus{\res{1}{1,0,1,1}}{7}\\
12. & \true & \then & \neg \newpos{1,p} \lor \neg \newpos{1,t_1} & \jus{NEC1}{11,6,10} \\
13. & \true & \then & \neg p \lor \newpos{1,p} & \jus{\res{1}{0,1,0,0}}{8} \\
14. & \true & \then & \neg p \lor \neg \newpos{1,t_1} & \jus{LRES}{13,12}\\
15. & \true & \then & \neg t_0 \lor \neg p & \jus{LRES}{14,9} \\
16. & \true & \then & \neg t_0 & \jus{LRES}{15,2} \\
17. & \start & \then & \false & \jus{IRES1}{16,1}\\
\end{array}
$
}

\noindent Clause (11) results from applying the Euclidean inference rule to clause (7). Clause (13) results from applying the reflexive inference rule to (8). The remaining clauses are derived by the resolution calculus for \system{K}{(1)}{}. As a contradiction is found, given by clause (17), the set of clauses is unsatisfiable and the original formula $\varphi$ is valid.
\end{example}

\section{Closing Remarks}

We have presented a sound, complete, and terminating proof method for logics of confluence, that is, normal multimodal systems where axioms of the form 

\vspace{-1ex}\begin{center}
$\G{\gp,\gq,\gr,\gs}{\agent}=\pos{\agent}^\gp\nec{\agent}^\gq \varphi \then \nec{\agent}^\gr\pos{\agent}^\gs\varphi$
\end{center}

\vspace{-1ex}\noindent where $\gp,\gq,\gr,\gs \in \Set{0,1}$, are valid. The axioms \G{\gp,\gq,\gr,\gs}{a} provide a general form for axioms widely used in logical formalisms applied to representation and reasoning within Computer Science. 

We have proved completeness of the proof method presented in this paper for eight families of logics and their fusions. The inference rules for particular instances of these logics can be systematically obtained and the resulting calculus can be implemented by adding to the existing prover for \system{K}{(n)}{} \cite{George2013} the clauses dependent on the clause-set. Efficiency, of course, depends on several aspects. Firstly, for certain classes of problems, dedicated proof methods might be more efficient. For instance, if the satisfiability problem for a particular logic is in {\sf NP} (as in the case of \system{S5}{(1)}{}), then our procedure may be less efficient as the satisfiability problem for~\system{K}{(1)}{} is already {\sf PSPACE}-complete \cite{Ladner77}. Secondly, efficiency might depend on the inference rules chosen to produce proofs for a specific logic. For instance, for \system{S5}{(n)}{}, the user can choose the inference rules related to reflexivity and Euclideanness, or choose the inference rules related to seriality, symmetry, and Euclideanness. The number of inference rules used to test the unsatisfiability of a set of clauses for a particular logic might affect the number of clauses generated by the resolution method as well as the size of the proof. As in the case of derived inference rules in other proof methods, using more inference rules might lead to shorter proofs. Thirdly, as in the case of the resolution-based method for propositional logic, efficiency might be affected by strategies used to search for a proof. Future work includes the design of strategies for \calculus{K}{(n)}{} and for specific logics of confluence. Fourthly, efficiency might also depend on the form of the input problem. For instance, comparisons between tableaux methods and resolution methods \cite{HustadtSchmidt02b,GTW2011} have shown that there is no overall better approach: for some problems resolution proof methods behave better, for others tableaux based methods behave better. Providing a resolution-based method for the logics axiomatised by \axiomname{K$_\agent$} and \G{\gp,\gq,\gr,\gs}{\agent} gives the user a choice for automated tools that can be used depending on the type of the input formulae.

There are quite a few dedicated methods for the logics presented in this paper. In general, however, those methods do not provide a systematic way of dealing with logics based on similar axioms or their extensions. Therefore, we restrict attention here to methods related to logics of confluence. Tableaux methods for logics of confluence where the mono-modal axioms \axiomname{T}, \axiomname{D}, \axiomname{B}, \axiomname{4}, \axiomname{5}, \axiomname{De} (for density, the converse of \axiomname{4}), and \axiomname{G} are valid, can be found in \cite{CastilhoCGH97,CerroG99}. For each of those axioms, a tableau inference rule is given. The inference rules can then be combined in order to provide proof methods for modal logics under \system{S5}{(1)}{}. Whilst the tableaux procedures in \cite{CastilhoCGH97,CerroG99} are designed for mono-modal logics they seem to be extandable to multimodal logics as long as there are no interactions between modalities. Those procedures do not cover all the logics investigated in this paper. In \cite{BasinMV97}, labelled tableaux are given for the mono-modal logics axiomatised by \axiomname{K} and axioms \G{\gp,\gq,\gr,\gs}{} where $\gq=\gs=0$ implies $\gp=\gr=0$. This restriction avoids the introduction of inference rules related to the identity predicate, but also excludes, for instance, functional and modally banal systems, which are treated by the method introduced in the present paper. In \cite{Blackburn02beyondpure}, hybrid logic tableaux methods for logics of confluence are given: the inference rules create nodes, labelled by nominals. The nominals are used in order to eliminate the Skolem function related to the existential quantifier in the first-order sentence corresponding to the axiom \G{\gp,\gq,\gr,\gs}{\agent}. This proof method provides tableau rules for all instances of the axiom. Soundness and completeness are discussed, but termination of the method is not dealt with and it is not clear what are the bounds for creating new nodes in the general case. In \cite{GPL:2011:display:tense}, sound, complete, and terminating display calculi for tense logics and some of its extensions, including those with the axiom  \G{\gp,\gq,\gr,\gs}{\agent}, are presented. It has been shown that these calculi have the property of \emph{separation}, that is, they provide complete proof methods for the component fragments. The paper investigates the relation between the display calculi and deep inference systems (where the sequent rules can be applied at any node of a proof tree). By finding appropriate propagation rules for the fusion of tense logic with either \system{S4}{(1)}{}, \system{S5}{(1)}{}, or functional systems, completeness of search strategies are presented. However, propagation rules for the axiom of convergence, $\axiomname{G1}$, or for the combination of path axioms (i.e.\ axioms of the form $\pos{}^i\varphi \then \pos{}^j \varphi$) with seriality are not given. Also related, in \cite{Baldoni98atableau}, prefixed tableaux procedures for confluence logics that validate the multimodal version of the axiom $\pos{a}\nec{b} \varphi \then \nec{c}\pos{d}\varphi$, where $\varphi$ is a formula, are given. Note that the logics in \cite{Baldoni98atableau} are systems with instances of the axiom \axiomname{\G{1,1,1,1}{a,b,c,d}}, that is, a logic which allows the interaction of the agents $a,b,c,d \in \Agents$, and might lead to undecidable systems.

To the best of our knowledge, there are no resolution-based proof methods for logics of confluence. However, resolution-based methods for modal logics, based on translation into first-order logic, have been proposed for several modal logics. A survey on translation-based approaches for non-transitive modal logics (i.e.\ modal logics that do not include the axiom \axiomname{4}) can be found in \cite{Nivelle00resolution-basedmethods:2000}. The translation-based approach has the clear advantage of being easily implemented, making use of well-established theorem-provers, and dealing with any logic that can be embedded into first-order, should it be decidable or not. However, first-order provers cannot deal easily with logics that embed some properties which are covered by particular axioms of confluence (e.g. functionality). In order to avoid such problematic fragments within first-order logic, the axiomatic translation principle for modal logic, introduced in \cite{SchmidtHustadt07a}, besides using the standard translation of a modal formulae into first-order, takes an axiomatisation for a particular modal logic and introduces a set of first-order modal axioms in the form of \emph{schema clauses}. As an example, adapted from \cite{SchmidtHustadt07a}, in order to prove that $\nec{\agent}\neg\nec{\agent}p$ is satisfiable in \system{KT4}{(n)}{}, for each modal subformula (i.e.\ $\nec{\agent}\neg\nec{\agent}p$ and $\nec{\agent}p$) and for each considered axiom (i.e.\ \axiomname{T} and \axiomname{4}), one schema clause is added, resulting in:

\vspace{-.5ex}{\small\begin{center}
$
\begin{array}{c}
\neg Q_{\nec{\agent}\neg\nec{\agent}p}(x) \lor \neg R(x,y) \lor Q_{\nec{\agent}\neg\nec{\agent}p}(y) \\
\neg Q_{\nec{\agent}p}(x) \lor \neg R(x,y) \lor Q_{\nec{\agent}p}(y)
\end{array}
$\quad\quad
$
\begin{array}{c}
\neg Q_{\nec{\agent}\neg\nec{\agent}p}(x) \lor Q_{\neg\nec{\agent}p}(y) \\
\neg Q_{\nec{\agent}p}(x) \lor Q_{p}(y)
\end{array}
$
\end{center}}

\vspace{-.3ex}\noindent where the predicate $Q_{\varphi}(x)$ can be read as $\varphi$ holds at world $x$ and $R$ is the predicate symbol to express the accessibility relation for agent $\agent$. Note that the clauses on the left are related to transitivity (\axiomname{4}) and the two clauses on the right are related to reflexivity (\axiomname{T}). The axiomatic translation approach is similar to the approach taken in the present paper and in \cite{ND07:jaal} as the schema clauses provide a way of talking about properties of the accessibility relation. As in our case, soundness follows easily from the properties of the translation. Termination follows from the fact that only a finite number of schema clauses are needed. However, as in the case of the proof method presented here, general completeness of the method is difficult to be proved and it is given only for particular families of logics. In \cite{DN2005}, a translation-based approach for properties which can be expressed by regular grammar logics (including transitivity and Euclideaness) is given. Completeness of the method has been proved for some families of logics.

In the present paper, we have restricted attention to the case where $\gp,\gq,\gr,\gs \in \Set{0,1}$, but we believe that the proof method can be extended in a uniform way for dealing with the unsatisfiability problem for any values of $\gp,\gq,\gr$, and $\gs$, by adding inference rules of the following form:

\vspace{-.5ex}{\small\begin{center}
$
\begin{array}{lrcl}
{[\res{\agent}{\gp,\gq,\gr,\gs}]} & \universal(l & \then & \pos{\agent}^{\,\gp} \nec{\agent}^{\gr} l')\\ \cline{2-4}
                      & \universal(l & \then & \nec{\agent}^\gr\pos{\agent}^{\gs} l')
\end{array}
$
\end{center}}

\vspace{-.5ex}\noindent which requires search for clauses that correspond to the normal form of the premise and the introduction of as many new definition symbols as the number of modalities occurring in the conclusion. The inference rule \res{\agent}{\gp,\gq,\gr,\gs} is obviously sound, but we have yet to identify the restrictions on the number of new propositional symbols introduced by the method in order to ensure termination. Future work includes this extension, the complexity analysis, the implementation of the proof method, and practical comparisons with other methods.

\bibliographystyle{abbrv}
\bibliography{confluence,mtp}

\pagebreak


\end{document}